\documentclass[runningheads,a4paper]{llncs}

\usepackage{graphicx}
\usepackage{latexsym}
\usepackage[cmex10]{amsmath}
\usepackage[tight,footnotesize,nooneline]{subfigure}
\usepackage{theorem}
\usepackage{amssymb}

\usepackage{tikz} 

\usepackage{url}
\urldef{\mailsa}\path|nco@maths.lth.se| 
\newcommand{\keywords}[1]{\par\addvspace\baselineskip
\noindent\keywordname\enspace\ignorespaces#1}

\newcommand{\real}{\mathbf{R}}

\begin{document}

\mainmatter  % start of an individual contribution

% first the title is needed

\title{On the Taut String Interpretation of the One-dimensional Rudin-Osher-Fatemi Model:\\ A New Proof, a Fundamental Estimate\\ and Some Applications} %:\\ A New Proof and Some Applications

% a short form should be given in case it is too long for the running head
\titlerunning{On the Taut String Interpretation}

\author{Niels Chr. Overgaard}
\authorrunning{N. Chr. Overgaard}
% (feature abused for this document to repeat the title also on left hand pages)

% the affiliations are given next; don't give your e-mail address
% unless you accept that it will be published
\institute{Centre for Mathematical Sciences\\Lund University, Sweden\\
\mailsa}

%\url{http://www.somewhere.edu/math}
%
% NB: a more complex sample for affiliations and the mapping to the
% corresponding authors can be found in the file "llncs.dem"
% (search for the string "\mainmatter" where a contribution starts).
% "llncs.dem" accompanies the document class "llncs.cls".
%

\toctitle{Lecture Notes in Computer Science}
\tocauthor{Authors' Instructions}
\maketitle

\begin{abstract}
A new proof of the equivalence of the Taut String Algorithm and the one-dimensional Rudin-Osher-Fatemi model is presented. Based on duality and the projection theorem in Hilbert space, the proof is strictly elementary. Existence and uniqueness of solutions to both denoising models follow as by-products. The standard convergence properties of the denoised signal, as the regularizing parameter tends to zero, are recalled and efficient proofs provided. Moreover, a new and fundamental bound on the denoised signal is derived. This bound implies, among other things, the strong convergence (in the space of functions of bounded variation) of the denoised signal to the insignal as the regularization parameter vanishes. The methods developed in the paper can be modified to cover other interesting applications such as isotonic regression.

\end{abstract}

\section{Introduction}\label{sec:intro}

In 2017 it is 25 years ago Leonid Rudin, Stanley Osher and Emad Fatemi proposed their now classical model for edge-preserving denoising of images~\cite{ROF-1992}. The present paper will investigate the properties of the one-dimensional version of the Rudin-Osher-Fatemi (ROF) model: To a given (noisy) signal $f\in L^2(I)$, defined on a bounded interval $I=(a,b)$, associate the (ROF) functional
\[
E_\lambda(u) = \lambda\int_a^b | u'(x)|\,dx + \frac{1}{2}\int_a^b(f(x)-u(x))^2\,dx\;,
\]
where $\lambda>0$ is a parameter. Define the denoised signal as the function $u_\lambda\in BV(I)$ which minimizes this energy, i.e.,
\begin{equation}\label{eq:ROFdef}
u_\lambda := \operatorname*{arg\,min}_{u\in BV(I)} E_\lambda(u)\;.
\end{equation}
The first term in the ROF-functional is the {\em total variation} $\int_a^b | u' |\,dx$ of the function $u$ multiplied by the positive weight $\lambda$, and $BV(I)$ denotes the set of functions on $I$ with finite total variation. Precise definitions will be given below.

The one-dimensional ROF model will compared to the {\em Taut string algorithm}, which is an alternative method for denoising of signals with applications in statistics, non-parametric estimation, real-time communication systems and stochastic analysis. The taut string algorithm has been extensively studied in the discrete setting by Mammen and van de Geer~\cite{Mammen-deGeer-1997}, Davies and Kovac~\cite{Davies-Kovac-2001} and by D\"{u}mbgen and Kovac~\cite{Dumbgen-Kovacs-2009}. Very recently, using methods from interpolation theory (Peetre's $K$-functional and the notion of invariant $K$-minimal sets), Setterqvist~\cite{Setterqvist-thesis-2016} has investigated the limits to which taut string methods may be extended. In the continuous setting, for analogue signals, the Taut string algorithm can be stated in the following manner: (Illustrated in Fig.~\ref{fig:TSA}.)

\begin{center}
	\begin{minipage}{1.0\textwidth}
		\rule[1mm]{1.0\textwidth}{.25mm}\\
		{\bf  The Taut String Algorithm}\\
		\rule[2mm]{1.0\textwidth}{.5mm}
		
		\noindent
		{\sc Input:} A bounded interval $I=(a,b)$, a (noisy) signal $f\in L^2(I)$ and a parameter $\lambda >0$.
		
		\noindent
		{\sc Output:} The denoised signal $f_\lambda\in L^2(I)$.
		
		\noindent
		{\sc Step 1.} Compute the cumulative signal,
		\[
		F(x) = \int_a^x f(t)\,dt\;,\quad x\in \overline{I}=[a,b]\,.
		\]
		
		\noindent
		{\sc Step 2.} Set 
		\begin{multline*}
		T_\lambda =\Big\{  W\in H^1(I)\,:\,  W(a)=F(a),\, W(b)=F(b),
		\text{ and } F-\lambda\leq W\leq F+\lambda  \,\Big\}\;.
		\end{multline*}
		(Graphically, this is the set of weakly differentiable $L^2$-functions with $L^2$-derivatives whose graphs lie within a tube around $F$ with the width $\lambda$.)
		
		\noindent
		{\sc Step 3.} Compute the unique minimizer $W_\lambda\in T_\lambda$ of the energy
		\begin{equation}\label{eq:energy}
		\min_{W\in T_\lambda} E(W):=\frac{1}{2}\int_a^b W'(x)^2\,dx\;.\quad\text{(`Taut string')}
		\end{equation}
		
		\noindent
		{\sc Step 4.} Set $f_\lambda = W_\lambda'$ (distributional derivative.)
		
		\noindent
		{\sc End.}
		
		\rule{1.0\textwidth}{.25mm}\\
	\end{minipage}
\end{center}

In its original formulation, the Taut string algorithm instruct us to find the solution of the shortest path problem
\begin{equation}\label{eq:SPP}
	\min_{W\in T_\lambda} L(W):=\int_a^b \sqrt{1+W'(x)^2}\,dx\;,
\end{equation}
hence the epithet `taut string'. However, the `stretched rubber band'-energy $E$ in step 3 of the algorithm is not only easier to handle analytically, it also has precisely the same solution as (\ref{eq:SPP}). While this is intuitively clear from our everyday experience with rubber bands and strings, the assertion is, mathematically speaking, not equally self-evident so a proof is offered in Appendix~\ref{sec:appendix}. 

The main purpose of this paper, the first of two, is to present a new, elementary proof of the following remarkable result:

\begin{theorem}\label{thm:equivalence}
	The Taut string algorithm and the ROF model yield the same solution; $f_\lambda = u_\lambda$.
\end{theorem}
This is not new; a discrete version of this theorem was proved in \cite{Mammen-deGeer-1997} and in \cite{Davies-Kovac-2001}. In the continuum setting, the equivalence result was explicitly stated and proved by Grasmair~\cite{Grassmair-2007}. There is also an extensive treatment in the book by Scherzer et al.~\cite[Ch. 4]{Scherzer-etal-2009}. Indeed, a few years earlier, Hinterm\"{u}ller and Kunisch~\cite[p.7]{Hintermuller-Kunisch-2004} refer, in a brief (but inconclusive) remark, to the close relation between the ROF model and the Taut string algorithm.

\begin{figure*}
	\centering
	\subfigure[][The input signal {$f$}.]{
		\begin{tikzpicture}[scale=.775,>=stealth]
		\draw [thick,->] (-1,0) -- (5,0) node[anchor=north] {\small $x$};
		\draw [thick,->] (-.5,-1.5) -- (-.5,2.5) node[anchor=west] {\small $y$};
		\draw [thick] (0,1.5) -- (.95,1.5);
		\draw [fill] (0,1.5) circle (.05);
		\draw (1,1.5) circle (.05);
		\draw [thick] (1,-1) -- (1.95,-1);
		\draw [fill] (1,-1) circle (.05);
		\draw (2,-1) circle (.05);
		\draw [thick] (2,-.5) -- (2.95,-.5);
		\draw [fill] (2,-.5) circle (.05);
		\draw (3,-.5) circle (.05);
		\draw [thick] (3,1) -- (4,1);
		\draw [fill] (3,1) circle (.05);
		\draw [fill] (4,1) circle (.05) node [anchor=south west,orange]{\small $f$};
		\draw [thick] (0,.1) -- (0,-.1) node[anchor=north, orange] {\small $a$};
		\draw [thick] (1,.1) -- (1,-.1);
		\draw [thick] (2,.1) -- (2,-.1);
		\draw [thick] (3,.1) -- (3,-.1);
		\draw [thick] (4,.1) -- (4,-.1) node[anchor=north, orange] {\small $b$};
		\end{tikzpicture}
	}
	\qquad\qquad
	\subfigure[][The cumulative signal {$F$} and the tube $T_\lambda$.]{
		\begin{tikzpicture}[scale=.775,>=stealth]
		\draw [thick,->] (-1,0) -- (5,0) node[anchor=north] {\small $x$};
		\draw [thick,->] (-.5,-1.5) -- (-.5,2.5) node[anchor=west] {\small $y$};
		\draw [thick] (-.4,.5) -- (-.6,.5) node[anchor=east,orange]{\small $\lambda$}; 
		\draw [thick] (-.4,-.5) -- (-.6,-.5) node[anchor=east,orange]{\small $-\lambda$};
		\draw [thick] (0,0) -- (1,1.5) -- (2,.5) -- (3,0) -- (4,1);
		\draw [fill] (0,0) circle (.05);
		\draw [fill] (4,1) circle (.05) node [anchor=south west,orange]{\small $F$};
		
		\draw [thick,lightgray] (0,.5) -- (1,2) -- (2,1) -- (3,.5) -- (4,1.5);
		\draw [fill,lightgray] (0,.5) circle (.05);
		\draw [fill,lightgray] (4,1.5) circle (.05) node [anchor=south west,orange]{\small $F+\lambda$};
		
		\draw [thick,lightgray] (0,-.5) -- (1,1) -- (2,0) -- (3,-.5) -- (4,.5);
		\draw [fill,lightgray] (0,-.5) circle (.05);
		\draw [fill,lightgray] (4,.5) circle (.05) node [anchor=west,orange]{\small $F-\lambda$};
		
		\draw [thick] (0,.1) -- (0,-.1);
		\draw [thick] (1,.1) -- (1,-.1);
		\draw [thick] (2,.1) -- (2,-.1);
		\draw [thick] (3,.1) -- (3,-.1);
		\draw [thick] (4,.1) -- (4,-.1) node[anchor=north] {\small $b$};
		\end{tikzpicture}
	}
	\\
	\subfigure[][The taut string $W_\lambda$.]{
		\begin{tikzpicture}[scale=.775,>=stealth]
		\draw [thick,->] (-1,0) -- (5,0) node[anchor=north] {\small $x$};
		\draw [thick,->] (-.5,-1.5) -- (-.5,2.5) node[anchor=west] {\small $y$};
		\draw [thick] (0,.1) -- (0,-.1);
		\draw [thick] (1,.1) -- (1,-.1);
		\draw [thick] (2,.1) -- (2,-.1);
		\draw [thick] (3,.1) -- (3,-.1);
		\draw [thick] (4,.1) -- (4,-.1) node[anchor=north]{\small $b$};
		
		\draw [thin] (0,0) -- (1,1.5) -- (2,.5) -- (3,0) -- (4,1);

		\draw [thick,lightgray] (0,.5) -- (1,2) -- (2,1) -- (3,.5) -- (4,1.5);
		\draw [fill,lightgray] (0,.5) circle (.05);
		\draw [fill,lightgray] (4,1.5) circle (.05);
		
		\draw [thick,lightgray] (0,-.5) -- (1,1) -- (2,0) -- (3,-.5) -- (4,.5);
		\draw [fill,lightgray] (0,-.5) circle (.05);
		\draw [fill,lightgray] (4,.5) circle (.05);
		
		\draw [thick,orange] (0,0) -- (1,1) -- (3,.5) -- (4,1);
		\draw [fill,orange] (0,0) circle (.05);
		\draw [fill,orange] (4,1) circle (.05) node[anchor=west,orange]{\small $W_\lambda$};
		\end{tikzpicture}
	}
	\qquad\qquad
	\subfigure[][The denoised signal {$f_\lambda$} together with the input signal.]{
		\begin{tikzpicture}[scale=.775,>=stealth]
		\draw [thick,->] (-1,0) -- (5,0) node[anchor=north] {\small $x$};
		\draw [thick,->] (-.5,-1.5) -- (-.5,2.5) node[anchor=west] {\small $y$};
		
		\draw [thick] (0,1.5) -- (.95,1.5);
		\draw [fill] (0,1.5) circle (.05);
		\draw (1,1.5) circle (.05);
		
		\draw [thick,orange] (0,1) -- (.95,1);
		\draw [fill,orange] (0,1) circle (.05);
		\draw [orange] (1,1) circle (.05);
		%---
		\draw [thick] (1,-1) -- (1.95,-1);
		\draw [fill] (1,-1) circle (.05);
		\draw (2,-1) circle (.05);
		
		\draw [thick] (2,-.5) -- (2.95,-.5);
		\draw [fill] (2,-.5) circle (.05);
		\draw (3,-.5) circle (.05);
		
		\draw [thick,orange] (1,-.25) -- (2.95,-.25);
		\draw [fill,orange] (1,-.25) circle (.05);
		\draw [orange] (3,-.25) circle (.05);
		%----
		\draw [thick] (3,1) -- (4,1);
		\draw [fill] (3,1) circle (.05);
		\draw [fill] (4,1) circle (.05) node [anchor=south west]{\small $f=F'$};
		
		\draw [thick,orange] (3,.5) -- (4,.5);
		\draw [fill,orange] (3,.5) circle (.05);
		\draw [fill,orange] (4,.5) circle (.05) node [anchor=west,orange]{\small $f_\lambda=W_\lambda'$};
		
		\draw [thick] (0,.1) -- (0,-.1) node[anchor=north] {\small $a$};
		\draw [thick] (1,.1) -- (1,-.1);
		\draw [thick] (2,.1) -- (2,-.1);
		\draw [thick] (3,.1) -- (3,-.1);
		\draw [thick] (4,.1) -- (4,-.1) node[anchor=north] {\small $b$};
		\end{tikzpicture}
	}
	\caption{A graphical illustrations of the steps in the Taut string algorithm applied to a piecewise constant signal.}\label{fig:TSA}
\end{figure*}
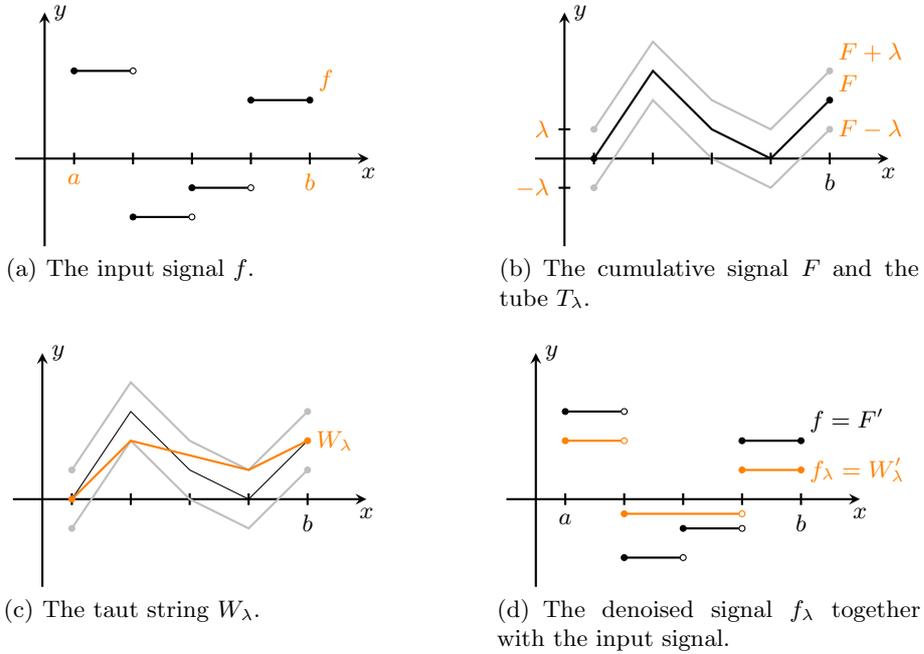

The second main result of the paper, whose proof we give in Section~\ref{sec:fundamental-estimate}, is the following ``fundamental'' estimate on the denoised signal:

\begin{theorem}
	\label{thm:L-S-consequence}
	If the signal $f$ belongs to $BV(I)$ then, for any $\lambda > 0$, the denoised signal $u_\lambda$ satisfies the inequality
	\begin{equation}\label{eq:LS-consequence}
	-(f')^- \leq u_\lambda' \leq (f')^+ \, ,
	\end{equation}
	where $(f')^+$ and $(f')^-$ denote the positive and the negative variations, respectively, of $f'$ (distributional derivative).
\end{theorem}
Just like $f'$, the derivative $u_\lambda'$ is computed in the distributional sense and is, in general, a signed measure. Furthermore, recall that $(f')^+$ and $(f')^-$ are finite positive measures satisfying $f'=(f')^+ - (f')^-$, see e.g. Rudin~\cite[Sec. 6.6]{Rudin-1986}.
The proof of the theorem is based on (an extension of) an inequality of H. Lewi and G. Stampacchia~\cite{Lewy-Stampacchia-1970} and uses the Taut String-interpretation of the ROF model (Theorem~\ref{thm:equivalence}) in an essential way.

As a significant consequence of Theorem~\ref{thm:L-S-consequence} we find that if the insignal $f$ belongs to $BV(I)$ then $u_\lambda\to f$ strongly in $BV(I)$ as $\lambda\to 0+$ . The usual Moreau-Yosida approximation result, see e.g. \cite[Ch. 17]{Ambrosio-etal-2000}, only gives the weaker $u_\lambda\to f$ in $L^2(I)$ and $\int_I |u_\lambda'|\,dx\to \int_I|f'|\,dx$ as $\lambda$ tends to zero.

To summarize, the main contributions of the paper are: i) The new proof of the equivalence theorem, presented here with the general reader in mind. ii) Establishment of a fundamental estimate on the solution of the ROF model. iii) The re-derivation some known properties of the ROF model and proof of some precise results on the rate of convergence $u_\lambda\to f$ as $\lambda$ tends to zero (Propositions~\ref{thm:properties}--\ref{thm:L2-continuity})---collecting all such result in one place! iv) The proof of the strong convergence result mentioned above (Proposition~\ref{thm:strong-convergence}). v) A new and slick proof of the (known) fact that $u_\lambda$ is a semi-group with respect to $\lambda$ (Proposition~\ref{thm:semi-group}). vi) In the final section we indicated how our method of proof can be modified and applied to \emph{isotonic regression}.

\section{Our Analysis Toolbox}\label{sec:tools}
Throughout this paper $I$ denotes an open, bounded interval $(a,b)$, where $a<b$ are real numbers, and $\bar{I}=[a,b]$ is the corresponding closed interval.

$C_0^1(I)$ denotes the space of continuously differentiable (test-)functions $\xi:I\to\real$ with compact support in $I$.

For $1\leq p \leq \infty$, $L^p(I)$ denotes the Lebesgue space of measurable functions $f:I\to\real$ with finite $p$-norm; $\|f\|_p:= \big(\int_a^b |f(x)|^p\,dx\big)^{1/p} < \infty$, when $p$ is finite, and $\|f\|_\infty = \operatorname{ess\,sup}_{x\in I}|f(x)| <\infty$ when $p=\infty$. The space $L^2(I)$ is a Hilbert space with the inner product $
\langle f,g\rangle =\langle f,g\rangle_{L^2(I)} := \int_a^b f(x)g(x)\,dx$ and the corresponding norm $\|f\|:=\sqrt{\langle f,f\rangle_{L^2(I)}}=\|f\|_2$.

We are going to need the Sobolev spaces over $L^2$:
\[
H^1(I)=\big\{ u\in L^2(I)\,:\, u'\in L^2(I) \big\}\;,
\]
were $u'$ denotes the distributional derivative of $u$. This is a Hilbert space with inner product $\langle u,v\rangle_{H^1}:= \langle u,v\rangle + \langle u',v'\rangle$ and norm $\|u\|_{H^1} =( \|u'\|_2^2 + \|u\|_2^2 )^{1/2}$. Any $u\in H^1(I)$ can, after correction on a set of measure zero, be identified with a unique function in $C(\bar{I})$. In particular, a unique value $u(x)$ can be assigned to $u$ for every $x\in\bar{I}$. 

The following subspace of $H^1(I)$ plays an important role in our analysis:
\[
H_0^1(I) = \big\{ u\in H^1(I)\,:\, u(a)=0\text{ and } u(b)=0\, \big\}\;.
\]
Here $\langle u,v\rangle_{H_0^1(I)} := \int_a^b u'(x)v'(x)\,dx$
defines an inner product on $H_0^1(I)$ whose induced norm $\|u\|_{H_0^1(I)}=\|u'\|_2$ is equivalent to the norm inherited from $H^1(I)$.

Finally, let $H$ be a (general) real Hilbert space with inner product between $u,v\in H$ denoted by $\langle u,v \rangle$ and the corresponding norm $\| u\|=\sqrt{\langle u,u \rangle}$. The following result is standard, see Br\'{e}zis~\cite[Th\'{e}or\`{e}me V.2]{Brezis-1999}:

\begin{proposition}[Projection Theorem]\label{thm:projection}
	Let $K\subset H$ be a non-empty closed convex set. Then for every $\varphi\in H$ there exists a unique point $u\in K$ such that
	\[
	\| \varphi-u \| = \min_{v\in K}\| \varphi- v \|.
	\]
	Moreover, the minimizer $u$ is characterized by the following property:
	\[
	u\in K\quad\text{and}\quad \langle \varphi-u,v-u\rangle \leq 0,\text{ for all $v\in K$}.
	\]
	The point $u$ is called the projection of $\varphi$ onto $K$, and is denoted $u=P_K(\varphi)$.
\end{proposition}

\section{Precise Definition of the ROF Model}\label{sec:BV}

The expression $\int_I |u'|\,dx$ for the total variation, makes sense for $u\in H^1(I)$ but is otherwise merely a convenient symbol. A more general and precise definition is needed; one which works in the case when $u'$ does not exist in the classical sense. The standard way to define the total variation is via duality: For $u\in L^1(I)$ set
\begin{equation*}\label{eq:TVdef}
J(u) = \sup\Big\{ \int_a^bu(x)\xi'(x)\,dx\, :\, \xi\in C_0^1(I),\, \|\xi\|_\infty\leq 1\Big\}\; .
\end{equation*}
If $J(u)<\infty$, $u$ is said to be a {\em function of bounded variation} on $I$, and $J(u)$ is called the {\em total variation} of $u$ (using the same notation as \cite{Chambolle-2004}). The set of all integrable functions on $I$ of bounded variation is denoted $BV(I)$, that is, $BV(I) = \big\{u\in L^1(I)\, :\, J(u) <\infty \big\}$. This becomes a Banach space when equipped with the norm $\|u\|_{BV}:=J(u) + \|u\|_{L^1}$. Notice that, as already indicated, if $u\in H^1(I)$ then $J(u)=\int_I |u'|\,dx < \infty$, so $u\in BV(I)$.

Let us illustrate how the definition works for a function with a jump discontinuity:

\begin{example}\label{ex:first} Let $u(x)=\operatorname{sign}(x)$ on the interval $I=(-1,1)$. For any $\xi\in C_0^1(I)$, satisfying $|\xi(x)|\leq 1$ for all $x\in I$, we have
\[
\int_{-1}^{1} u(x)\xi'(x)\,dx = \int_0^1\xi'(x)\,dx -\int_{-1}^0\xi'(x)\,dx=-2\xi(0) \leq 2,
\]
where equality holds for any admissible $\xi$ which satisfies $\xi(0)=-1$. So
$J(u)=2$ and $u\in BV(I)$, as predicted by intuition.
\end{example}

In this example the supremum is attained by many choices of $\xi$. This is not always the case; if $u(x)=x$ on $I=(0,1)$ then $J(u)=1$, but the supremum is not attained by any admissible test function. 

The following lemma shows that the definition of the total variation $J$ and the space $BV(I)$ can be moved to a Hilbert space-setting involving $L^2$ and $H_0^1$. 

\begin{lemma}\label{thm:newTVdef}
	Every $u\in BV(I)$ belongs to $L^2(I)$ and 
	\begin{equation}\label{eq:newdef}
	J(u)=\sup_{\xi\in K}\, \langle u,\xi'\rangle_{L^2(I)}\;,
	\end{equation}
	where $K=\{\,\xi\in H_0^1(I)\,:\, \|\xi\|_\infty\leq 1 \,\}$, which is a closed and convex  set in $H_0^1(I)$.
\end{lemma}

\begin{proof}
	If $u\in BV(I)$ then Sobolev's lemma for functions of bounded variation, see \cite[p. 152]{Ambrosio-etal-2000}, ensures that $u\in L^\infty(I)$. This in turn implies $u\in L^2(I)$ because $I$ is bounded. The (ordinary) Sobolev's lemma asserts that $H_0^1(I)$ is continuously embedded in $L^\infty(I)$. Since $K$ is the inverse image under the embedding map of the unit ball in $L^\infty(I)$, which is both closed and convex, we draw the conclusion that $K$ is closed and convex in $H_0^1$.
	
	It only remains to prove (\ref{eq:newdef}). Clearly $J(u)$ cannot exceed the right hand side because the set $\{\, \xi\in C_0^1(I)\,:\, \|\xi\|_\infty\leq 1\,\}$ is contained in $K$. To verify that equality holds it is enough to prove the inequality
	\begin{equation}\label{eq:TVinequality}
	\langle u,\xi'\rangle_{L^2(I)} \leq J(u) \| \xi\|_{\infty},\quad\text{for all $\xi\in H_0^1(I)$},
	\end{equation}
	as it implies that the right hand side of (\ref{eq:newdef}) cannot exceed $J(u)$.
	To do this, we first notice that the inequality in holds for all $\zeta\in C_0^1(I)$. This follows by applying homogeneity to the definition of $J(u)$. Secondly, if $\xi\in H_0^1(I)$ we can use that $C_0^1(I)$ is dense in $H_0^1(I)$ and find functions $\zeta_n\in C_0^1(I)$ such that $\zeta_n\to \xi$ in $H_0^1(I)$ (and in $L^\infty(I)$ by the continuous embedding). It follows that
	\[
	\langle u,\xi'\rangle_{L^2(I)} = \lim_{n\to\infty} \langle u,\zeta_n'\rangle_{L^2(I)} \leq J(u) \lim_{n\to\infty} \| \zeta_n\|_{\infty} = J(u) \| \xi\|_{\infty},
	\]
	which establishes (\ref{eq:TVinequality}) and the proof is complete.\qed
\end{proof}

The inequality (\ref{eq:TVinequality}) combined with the Riesz representation theorem (cf. e.g. \cite[Thm. 1.54]{Ambrosio-etal-2000}) implies that the distributional derivative $u'$ of $u\in BV(I)$ is a signed (Radon) measure $\mu$ on $I$, and that we may write $\langle u, \xi' \rangle_{L^2(I)} = \int_I \xi\,d\mu$. This will be useful later on.

We can now give the precise definition of the ROF model: For any $f\in L^2(I)$ and any real number $\lambda>0$ the ROF functional is the function $E_\lambda:BV(I)\to\real$ given by
\begin{equation}\label{eq:ROFenergy}
E_\lambda(u)=\lambda J(u)+\frac{1}{2}\| f- u \|_{L^2(I)}^2\;.
\end{equation}
Denoising according to the ROF model is the map $L^2(I)\ni f\mapsto u_\lambda \in BV(I)$ defined by (\ref{eq:ROFdef}). To emphasise the role of the in-signal $f$ we sometimes write $E_\lambda(f;u)$ instead of $E_\lambda(u)$. Well-posedness of the ROF model is demonstrated in the next section.

\section{Existence Theory for the ROF Model}\label{sec:main_result}

We begin with a simple observation: if $u\in BV(I)$  then $J(u+c) = J(u)$ for any real constant $c$. This property of the total variation has two important consequences. First of all, $E_\lambda(f;u)=E_\lambda(f-c;u-c)$ for any constant $c$. Taking $c$ to be the mean value of $f$ shows that we may assume, as we do throughout this paper, that the in-signal satisfies $\int_If\,dx=0$. This assumption implies that the cumulative signal $F(x)$ satisfies $F(a)=F(b)=0$, hence $F\in H_0^1(I)$. This plays an important role in our analysis.

Secondly, since $f$ has mean value zero, it is enough to minimize $E_\lambda$ over the subspace of $BV(I)$ consisting of functions with mean value zero. To see this, let $P$ be the orthogonal projection (in $L^2(I)$) onto this subspace. An easy computation yields the identity $E_\lambda(Pu) =E_\lambda(u)-\frac{1}{2}\|u-Pu\|^2$, which shows that $u$ can be a minimizer of $E_\lambda$ only if it belongs to the range of $P$.

The following result is the key theorem of our paper.

\begin{theorem}\label{thm:main}
	We have the equality
	\begin{equation}\label{eq:fundamental_id}
	\min_{u\in BV(I)} E_\lambda(u) = \max_{\xi\in K} \frac{1}{2}\Big\{ \| f\|_{L^2(I)}^2 - \| f-\lambda\xi'\|_{L^2(I)}^2\Big\}\;,
	\end{equation}
	with the minimum achieved by a unique $u_\lambda\in BV(I)$ and the maximum by a unique $\xi_\lambda\in K$. The two functions are related by the identity
	\begin{equation}\label{eq:fundamental-rel}
	u_\lambda=f-\lambda \xi_\lambda'\;,
	\end{equation}
	and satisfy
	\begin{equation}\label{eq:TV}
	J(u_\lambda) = \langle u_\lambda , \xi_\lambda' \rangle_{L^2(I)}\;.
	\end{equation}
	Moreover, if $u_\lambda\neq 0$, then $\|\xi_\lambda\|_\infty=1$. Conversely, if a pair of functions $\bar{u}\in BV(I)$ and $\bar{\xi}\in K$ satisfy both the condition in (\ref{eq:fundamental-rel}); $\bar{u}=f-\lambda \bar{\xi}'$, as well as (\ref{eq:TV}); $J(\bar{u}) = \langle \bar{u} , \bar{\xi'} \rangle_{L^2(I)}$, then $\bar{u}=u_\lambda$ and $\bar{\xi}=\xi_\lambda$. 
	%In particular $E_\lambda(\bar{u})=(1/2) ( \|f\|^2 - \| f-\lambda \bar{\xi}'\|^2)$, i.e., the equality (\ref{eq:fundamental_id}) holds.
\end{theorem}
This result is a special instance of the Fenchel-Rockafellar theorem, see e.g. \cite[p. 11]{Brezis-1999}. It is tailored with our specific needs in mind and will be proved with our bare hands using the projection theorem. The general version is used in Hinterm\"{u}ller and Kunisch~\cite{Hintermuller-Kunisch-2004} in their analysis of the multidimensional ROF model. Moreover, the equality (\ref{eq:fundamental_id}) has played an important role in the development of numerical algorithms for total variation minimization, both directly, as for instance in Zhu et al.~\cite{Zhu-etal-2007} or, indirectly, as in Chambolle~\cite{Chambolle-2004}.

Before the proof starts, let us remind the reader of the following general fact: If $M$ and $N$ are arbitrary non-empty sets and $\Phi:M\times N\to\mathbf{R}$ is any real valued function, then it is easy to check that
\begin{equation}\label{eq:infsup}
\inf_{x\in M}\sup_{y\in N}\Phi(x,y) \geq \sup_{y\in N}\inf_{x\in M}\Phi(x,y)\;,
\end{equation}
is always true. The use of $\inf$'s and $\sup$'s are important, as neither the greatest lower bounds nor the least upper bounds are necessarily attained.

\begin{proof}
	Since $E_\lambda(u) = \sup_{\xi\in K} \lambda\langle u,\xi'\rangle + \frac{1}{2}\|f-u\|^2$ it follows from (\ref{eq:infsup}) that
	\[
	\inf_{u\in BV(I)}E_\lambda(u) \geq \sup_{\xi\in K}\Big\{ \inf_{u\in BV(I)}\lambda \langle u,\xi'\rangle +\frac{1}{2}\|u-f\|^2 \Big\}\;.
	\]
	We first solve, for $\xi\in K$ fixed, the minimization problem on the right hand-side. Expanding $\|f-u\|^2$ and completing squares with respect to $u$ yields:
	\[
	\lambda \langle u,\xi'\rangle +\frac{1}{2}\|u-f\|^2 =
	\frac{1}{2}\Big\{ \| u-(f-\lambda\xi')\|^2 -\|f-\lambda\xi'\|^2+\|f\|^2 \Big\}
	\]
	The right hand-side is clearly minimized by the $L^2(I)$-function $u=f-\lambda\xi'$ and 
	\begin{equation}\label{eq:semi-id}
	\inf_{u\in BV(I)}E_\lambda(u) \geq \sup_{\xi\in K} \frac{1}{2}\Big\{ \|f\|^2-\|f-\lambda \xi'\|^2\Big\}
	\end{equation}
	holds. The maximization problem on the right hand side is equivalent to
	\begin{equation}\label{eq:closest_point}
	\inf_{\xi\in K} \|f-\lambda \xi'\| = \inf_{\xi\in K} \| F'-\lambda \xi'\| = \lambda \inf_{\xi\in K} \| \lambda^{-1}F -\xi\|_{H_0^1(I)}\;.
	\end{equation}
	By Proposition~\ref{thm:projection}, this problem has the unique solution $\xi_\lambda=P_K(\lambda^{-1}F)\in K$, so the supremum is attained in (\ref{eq:semi-id}). Now, let the function $u_\lambda$ be defined by  (\ref{eq:fundamental-rel}) in the theorem. A priori, $u_\lambda$ belongs to $L^2(I)$, but we are going to show that $u_\lambda\in BV(I)$: The characterization of $\xi_\lambda$ according in the projection theorem states that $
	\xi_\lambda\in K$ and  $\langle f-\lambda \xi_\lambda', \lambda\xi'-\lambda \xi_\lambda'\rangle\leq 0$ for all $\xi\in K$.
	If we use the definition of $u_\lambda$ and divide by $\lambda>0$ this characterization becomes
	\[
	\langle u_\lambda,\xi'\rangle \leq \langle u_\lambda,\xi_\lambda'\rangle\quad\text{for all $\xi\in K$,}
	\]
	where the right hand-side is finite. It follows from the definition of the total variation that $u_\lambda\in BV(I)$ with $J(u_\lambda)=\langle u_\lambda,\xi_\lambda'\rangle$, as asserted in the theorem. (This reasoning can be reversed; if (\ref{eq:TV}) is true then $\xi_\lambda$ is the minimizer in (\ref{eq:closest_point}).) Also, if $u_\lambda\neq 0$ then $\|\xi_\lambda\|_\infty <1$ is not consistent with the maximizing property (\ref{eq:TV}), hence $\|\xi_\lambda\|_\infty =1$, as claimed.
	
	It remains to be verified that $u_\lambda$ minimizes $E_\lambda$ and that equality holds in (\ref{eq:semi-id}). This follows from a direct calculation:
	\begin{align*}
	\inf_{u\in BV(I)} E_\lambda (u) &\geq \max_{\xi\in K} \frac{1}{2}\Big\{ \|f\|^2-\|f-\lambda \xi'\|^2\Big\} = \frac{1}{2}\|f\|^2 - \frac{1}{2}\|u_\lambda\|^2\\
	&= \frac{1}{2}\|f\|^2 +\frac{1}{2}\|u_\lambda\|^2 - \|u_\lambda\|^2 = \frac{1}{2}\|f\|^2 +\frac{1}{2}\|u_\lambda\|^2 -\langle u_\lambda,f-\lambda\xi_\lambda'\rangle\\
	&=\frac{1}{2}\|f-u_\lambda\|^2 + \langle u_\lambda, \lambda\xi_\lambda\rangle =\frac{1}{2}\|f-u_\lambda\|^2 + \lambda J(u_\lambda) = E_\lambda(u_\lambda)\;.
	\end{align*}
	So $\inf E_\lambda(u) = E_\lambda(u_\lambda)$, the infimum is attained, and equality holds in (\ref{eq:semi-id}). The inequality $E_\lambda(u)-E_\lambda(u_\lambda)\geq \frac{1}{2}\| u-u_\lambda \|^2$ implies the uniqueness of $u_\lambda$. The converse statement is proved by back-tracking the steps of the above proof.\qed
\end{proof}

Denoising is a non-expansive mapping:

\begin{corollary}\label{thm:non-expansiveness}
	If $f$ and $\tilde{f}$ are signals in $L^2(I)$ and the corresponding denoised signals are denoted $u_\lambda$ and $\tilde{u}_\lambda$, respectively, then $\| \tilde{u}_\lambda -u_\lambda \|_{L^2(I)} \leq \| \tilde{f} - f\|_{L^2(I)}$.
\end{corollary}
Like Theorem~\ref{thm:main} this a special instance of a more general result about Moreau-Yosida approximation (or of the proximal map), see \cite[Theorem 17.2.1]{Attouch-etal-2015}. However, the result is easily verified by the reader using the characterization of the ROF-minimzer given in the theorem.

The equivalence of the two denoising models can now be established:

\begin{proof}[of Theorem~\ref{thm:equivalence}]
It follows from Theorem~\ref{thm:main} that the minimizer $u_\lambda$ of the ROF functional is given by $u_\lambda = f- \lambda \xi_\lambda'$ where $\xi_\lambda$ is the unique solution of
\begin{equation}\label{eq:closest_to_f}
\min_{\xi\in K}\frac{1}{2}\| f-\lambda\xi'\|_{L^2(I)}^2\;.
\end{equation}
If we introduce the new variable $W:=F-\lambda \xi$, where $F\in H_0^1(I)$ is the cumulative signal, then $W\in H_0^1(I)$ and the condition $\|\xi\|_\infty\leq 1$ implies that $W$ satisfies $F(x)-\lambda\leq W(x)\leq F(x)+\lambda$ on $I$. Therefore (\ref{eq:closest_to_f}) is equivalent to $\min_{W\in T_\lambda}(1/2)\| W'\|_{L^2(I)}^2$,
which is the minimization problem in step 3 of the Taut string algorithm whose solution we denoted $W_\lambda$. It follows that $W_\lambda=F-\lambda\xi_\lambda$ and differentiation yields $f_\lambda = W_\lambda'= f-\lambda \xi_\lambda' = u_\lambda$,
the desired result.\qed
\end{proof}

It is interesting to note that Theorem~\ref{thm:main} associates a \emph{unique} test function $\xi_\lambda\in K$ with the solution $u_\lambda$ of the ROF model such that $J(u_\lambda)=\langle u_\lambda, \xi_\lambda'\rangle_{L^2}$, in particular if we compare to the situation in Example~\ref{ex:first}. A concrete case looks as follows:

\begin{example}
	Let $f(x)=\operatorname{sign}(x)$ be the step function defined on $I=(-1,1)$. An easy calculation, based on the Taut string interpretation, shows that if $0<\lambda<1$ then $u_\lambda=(1-\lambda)\operatorname{sign}(x)$ and $\xi_\lambda=|x|-1\in H_0^1(I)$. Notice that $\xi_\lambda$ is not in $C_0^1(I)$, so the extension of the space of test functions from $C_0^1$ to $H_0^1$ is essential to our theory. For $\lambda \geq 1$ we find $u_\lambda=0$ and $\xi_\lambda = \lambda^{-1}(|x|-1)$. Notice that $\|\xi_\lambda\|_\infty = 1$ when $u_\lambda\neq 0$.
\end{example}

Our proof of Theorem~\ref{thm:equivalence} is essentially a change of variables, and as such, is almost a `derivation' of the taut string interpretation. We also get the existence and uniqueness of solutions to both models in one stroke. By contrast, Grassmair's proof \cite{Grassmair-2007} shows that $u_\lambda$ and $W_\lambda'$ satisfy the same set of three necessary conditions, and that these conditions admit at most one solution. The point is driven home by establishing existence separately for both models. The argument assumes $f\in L^\infty$ and involves a fair amount of measure theoretic considerations. The proof of equivalence given in Scherzer et al.~\cite{Scherzer-etal-2009} is based on a thorough functional analytic study of Meyer's G-norm and is not elementary.

\section{Applications of the Taut String Interpretation}

We now show how some known, and some new, properties of the ROF model can be understood in the light of its equivalence to the Taut string algorithm. 

The Taut string algorithm suggests that $W_\lambda=0$, and therefore $u_\lambda=0$, when $\lambda$ is sufficiently large, and that $W_\lambda$ must touch the sides $F\pm\lambda$ of the tube $T_\lambda$ when $\lambda$ is small. These assertions can be made precise:

\begin{proposition}\label{thm:properties}
	(a) The denoised signal $u_\lambda = 0$ if and only if $\lambda \geq \|F\|_\infty$, and \\
	(b)  if $0 < \lambda < \|F\|_\infty$ then $\| F - W_\lambda\|_\infty = \lambda$.\\
	(c) $\|W_\lambda \|_\infty = \max(0, \|F\|_\infty - \lambda)$. 
\end{proposition}
The results (a) and (b) are well-known and proofs, valid in the multi-dimensional case, can be found in Meyer's treatise~\cite{Meyer-2000}. Notice that the maximum norm $\|F\|_\infty$ of the cumulative signal $F$ coincides, in one dimension, with the Meyer's G-norm $\|f\|_{\ast}$ of the signal $f$.
Theorem~\ref{thm:main} and the taut string interpretation of the ROF model allow us to give very short and direct proofs of all three properties.

\begin{proof}
	(a) By Theorem~\ref{thm:equivalence}, the denoised signal $u_\lambda$ is zero if and only if the taut string $W_\lambda$ is zero. We know that $W_\lambda=F-\lambda\xi_\lambda$ where, as seen from  (\ref{eq:closest_point}), $\xi_\lambda$ is the projection in $H_0^1(I)$ of $\lambda^{-1}F$ onto the closed convex set $K$. Therefore $u_\lambda = 0$ if and only if $\lambda^{-1}F\in K$, that is, if and only if $\| F\|_\infty \leq \lambda$, as claimed.
	
	(b) If $0 < \lambda < \|F\|_\infty$ then $u_\lambda\neq 0$ hence $\| \xi_\lambda\|_\infty = 1$, by Theorem~\ref{thm:main}. The assertion now follows by taking norms in the identity $\lambda \xi_\lambda = F-W_\lambda$.
	
	(c) The equality clearly holds when $\lambda\geq \|F\|_\infty$ because $W_\lambda=0$ by (a). When $c:=\|F\|_\infty-\lambda >0$ we use a truncation argument:  If $W\in T_\lambda$ then so does $\hat{W}:=\min(c,W)$, in particular $c>0$ ensures that $\hat{W}(a)=\hat{W}(b)=0$. Since $E(\hat{W})\leq E(W)$, and $W_\lambda$ is the (unique) minimizer of $E$ over $T_\lambda$, we conclude that $\max_I W_\lambda \leq c$. A similar argument gives $-\min_I W_\lambda\leq c$. Thus $\|W_\lambda \|_\infty \leq \max(0, \|F\|_\infty - \lambda)$. The reverse inequality follows from (b). \qed
\end{proof}

Now define, for $\lambda > 0$, the \emph{value function}
\[
e(\lambda) := \inf_{u\in BV(I)} E_\lambda(u),
\]
that is, $e(\lambda) = E_\lambda(u_\lambda)$. The next two theorems contains essentially well-known results.

\begin{proposition}\label{thm:value-function}
	The function $e:(0,+\infty)\to (0,+\infty)$ is nondecreasing and concave, hence continuous,  and satisfies 
	\[
	e(\lambda)=\|f\|^2/2 \quad \text{for $\lambda \geq \|F\|_\infty$ and}\quad \lim_{\lambda\to 0+} e(\lambda)=0. 
	\]
	In particular, if $f\in BV(I)$ then $e(\lambda)=O(\lambda)$ as $\lambda\to 0+$.
\end{proposition}

\begin{proof}
	If $\lambda_2\geq \lambda_1 > 0$ then the inequality $E_{\lambda_2}(u) \geq E_{\lambda_1}(u)$ holds trivially for all $u$. Taking infimum over the functions in $BV(I)$ yields $e(\lambda_2) \geq e(\lambda_1)$, so $e$ is nondecreasing. 
	
	For any $u$ the right hand side of the inequality
	\[
	e(\lambda) \leq E_\lambda(u) = \lambda J(u) + \frac{1}{2}\|u-f\|^2\;,
	\]
	is an affine, and therefore a concave, function of $\lambda$. Because the infimum of any family of concave functions is again concave, it follows that $e(\lambda) =  \inf_{u\in BV(I)} E_\lambda(u)$ is concave. 
	
	For $\lambda \geq \|F\|_\infty$ we know from the previous theorem that $u_\lambda = 0$, so $e(\lambda)=E_\lambda(0)=\|f\|^2/2$. 
	
	To prove the assertion about $e(\lambda)$ as $\lambda$ tends to zero from the right, we first assume that $f\in BV(I)$, in which case it follows that $0 < e(\lambda) \leq E_\lambda(f) =\lambda J(f)$, so $e(\lambda)=O(\lambda)$ because $J(f)<\infty$.
	
	If we merely have $f\in L^2(I)$ an approximation argument is needed: For any $\epsilon > 0$ take a function $f_\epsilon\in H_0^1(I)$ such that $\| f-f_\epsilon\|^2/2 < \epsilon$. Then $f_\epsilon\in BV(I)$ and
	$0\leq e(\lambda) \leq E_\lambda(f_\epsilon) < \lambda J(f_\epsilon) + \epsilon.$ It follows that $0 \leq \operatorname*{lim\,sup}_{\lambda\to 0+} e(\lambda) < \epsilon$. Since $\epsilon$ is arbitrary, we get $\lim_{\lambda\to 0+}e(\lambda)=0$.\qed
\end{proof}

The first part of next the proposition is a special instance of a much more general result, see Attouch et al.~\cite[Theorem 17.2.1]{Attouch-etal-2015}. The second part contains a quantification of the rate of convergence which is not easily located in the literature.

\begin{proposition}\label{thm:L2-continuity}
	For any $f\in L^2(I)$ we have $u_\lambda\to f$ in $L^2$ as $\lambda\to 0+$. Moreover, if $f\in BV(I)$ then $\|u_\lambda-f\|_{L^2(I)}=o(\lambda^{1/2})$ and $J(u_\lambda)\to J(f)$ as $\lambda\to 0+$.
\end{proposition}

\begin{proof}
	The obvious inequality $\| f-u_\lambda\|^2/2\leq e(\lambda)$ and the fact $\lim_{\lambda\to 0+}e(\lambda)=0$, proved above, implies the first assertion. When $f\in BV(I)$ it follows from the inequality $\lambda J(u_\lambda) + \frac{1}{2}\|u_\lambda-f\|_{L^2(I)}^2 = e(\lambda)\leq E_\lambda(f) =\lambda J(f)$ that
	\begin{equation}\label{eq:an_inequality}
	\|u_\lambda-f\|_{L^2(I)}^2 \leq 2\lambda( J(f)-J(u_\lambda))\;.
	\end{equation}
	Consequently $\|u_\lambda-f\|^2_{L^2(I)} = O(\lambda)$ and $J(u_\lambda)\leq J(f)$ for all $\lambda>0$. But we can do slightly better than that. Since $u_\lambda\to f$ in $L^2$ as $\lambda\to 0+$, we get $J(f)\leq \operatorname*{lim\,inf}_{\lambda\to 0+} J(u_\lambda)$, by the lower semi-continuity of the total variation $J$, cf. \cite{Ambrosio-etal-2000}. Since $J(u_\lambda)\leq J(f)$ we also obtain an estimate from below: $\operatorname*{lim\,sup}_{\lambda\to 0+} J(u_\lambda) \leq J(f)$. We conclude that $\lim_{\lambda\to 0+}J(u_\lambda) = J(f)$. If this is used in (\ref{eq:an_inequality}) we find that $\|u-f\|^2_{L^2(I)} = o(\lambda)\text{ as $\lambda\to 0+$}$.\qed
\end{proof}

\section{Proof and Applications of the Fundamental Estimate}\label{sec:fundamental-estimate}

We begin with the proof of our estimate on the derivative of the denoised signal:

\begin{proof}[of Theorem~\ref{thm:L-S-consequence}]
	This estimate is a consequence of the extension to bilateral obstacle problems of the original Lewy-Stampacchia inequality~\cite{Lewy-Stampacchia-1970}. The bilateral obstacle problem, in the one-dimensional setting, is to minimize the energy $E(u):=\frac{1}{2}\int_a^b u'(x)^2\,dx$ in (\ref{eq:energy}) over the closed convex set $C=\{ u\in H_0^1(I): \phi(x) \leq u(x)\leq \psi(x) \text{ a.e. $I$}\}$. The obstacles are functions $\phi,\psi\in H^1(I)$ which satisfy the conditions $\phi <\psi$ on $I$, and $\phi < 0 <\psi$ on $\partial I =\{ a,b\}$. This ensures that $C$ is nonempty.
	
	Suppose $\phi'$ and $\psi'$ are in $BV(I)$, such that $\phi''$ and $\psi''$ are signed measures, then the solution $u_0$ of $\min_{u\in C} E(u)$ satisfies the following inequality (as measures)
	\begin{equation}\label{eq:LS-ineq}
		-(\phi'')^- \leq u_0'' \leq (\psi'')^+\, .
	\end{equation}
	Here the notation $\mu^+$ and $\mu^-$ is used to denote the positive and negative variation, respectively, of a signed measure $\mu$. This is the generalization of the Lewy-Stampacchia inequality. An abstract proof, valid in a much more general setting, can be found in Gigli and Mosconi~\cite{Gigli-Mosconi-2015}. The assumption of our theorem, that $f\in BV(I)$, implies that $F''=f'$ is a signed measure. If we apply (\ref{eq:LS-ineq}) with $\phi=F-\lambda$ and $\psi = F+\lambda$ then we find that the taut string $W_\lambda$ satisfies
	\[
	-(F'')^- = -((F-\lambda)'')^- \leq W_\lambda'' \leq ((F+\lambda)'')^+ = (F'')^+ \, .
	\]
	The fundamental estimate (\ref{eq:LS-consequence}) follows if we substitute the identities $F'=f$ and $W_\lambda'=u_\lambda$ into the above inequality.\qed
\end{proof}

Having established Theorem~\ref{thm:L-S-consequence} we are able to prove the following result about the strong convergence in $BV(I)$ of the ROF-minimizer as the regularization weight approaches zero.

\begin{proposition}\label{thm:strong-convergence}
	If $f\in BV(I)$ then 
	\[
	J(f-u_\lambda)=J(f)-J(u_\lambda).
	\]
	In particular, both $J(f-u_\lambda)$ and $\| f-u_\lambda\|_{BV}$ tend to zero as $\lambda\to 0+$.
\end{proposition}

\begin{proof}
The measures $(f')^+$ and $(f')^-$ are concentrated on disjoint measurable sets (Hahn decomposition, see \cite[Sec. 6.14]{Rudin-1986}), so Proposition~\ref{thm:L-S-consequence} implies the pair of inequalities, $0\leq (u_\lambda')^+\leq (f')^+$ and $0\leq (u_\lambda')^-\leq (f')^-$. A direct calculation, using the fact that $J(v)=(v')^+(I) + (v')^-(I)$ for any function $v\in BV(I)$, yields
\begin{align*}
J(f-u_\lambda) &= (f'-u_\lambda')^+(I) +(f'-u_\lambda')^-(I)\\[1mm]
&=(f')^+(I) - (u_\lambda')^+(I) + (f')^-(I) - (u_\lambda)^-(I)\\[1mm]
&= J(f)-J(u_\lambda),
\end{align*}
where the right hand-side tends to zero as $\lambda\to 0+$, by Theorem~\ref{thm:L2-continuity}.\qed
\end{proof}

Theorem~\ref{thm:L-S-consequence} also implies the first part of the following

\begin{proposition}
	If $f$ is piecewise constant function on $I$, then so is $u_\lambda$ for all $\lambda >0$. Moreover, there exists a number $\bar{\lambda}>0$ and a piecewise linear function $\bar{\xi}\in K$ such that $\xi_\lambda=\bar{\xi}$ for all $\lambda$, $0<\lambda \leq \bar{\lambda}$.
\end{proposition}
	
The latter half of the proposition can be proved using the characterization of solutions in Theorem~\ref{thm:main}. We mention this result (but omit its proof---the first part being easy, the second, somewhat lengthy) because it implies what is possibly the strongest imaginable approximation result:

\begin{proposition}
	If $f$ is piecewise constant, then $\| f -u_\lambda \|_{L^2(I)} = O(\lambda)$, $\lambda\to 0+$. %(And $\| f-u_\lambda\|_\infty \leq (\lambda/4)\| \bar{\xi}' \|^2= o(\lambda)$ as $\lambda\to 0+$.)
\end{proposition}

\begin{proof}
	We know from (\ref{eq:an_inequality}) that $(1/2)\| f - u_\lambda \|^2_{L^2(I)} \leq \lambda ( J(f) - J(u_\lambda))$ so an estimate of the difference $J(f)-J(u_\lambda)$ is needed. By Theorem~\ref{thm:main}, $J(u_\lambda) = \langle u_\lambda, \xi_\lambda'\rangle_{L^2(I)}$. Since $\xi_\lambda = \bar{\xi}$ when $\lambda$ is close to zero it follows that 
	\[
	J(f) = \lim_{\lambda\to 0+} J(u_\lambda) = \lim_{\lambda\to 0+} \langle u_\lambda, \bar{\xi}'\rangle_{L^2(I)} = \langle f, \bar{\xi}'\rangle_{L^2(I)}.
	\]
	Computing the scalar product of $u_\lambda = f-\lambda \xi_\lambda'$ and $\xi_\lambda'=\bar{\xi}'$ yields $J(u_\lambda) = \langle u_\lambda , \bar{\xi}'\rangle = \langle f-\lambda \bar{\xi}', \bar{\xi}'\rangle = J(f) -\lambda \| \bar{\xi}' \|^2$. Hence $J(f) -J(u_\lambda) = O(\lambda)$, $\lambda\to 0+$.\qed	
\end{proof}

Our interest in the various limits as $\lambda\to 0+$ is motivated by the fact that $\lambda \mapsto u_\lambda$ is a semi-group; statements about limits at $\lambda=0$ can be translated to limits at any $\lambda >0$. 

\begin{proposition}[Semi-group property]\label{thm:semi-group}
	Let $f\in L^2(I)$. With the convention (mentioned above) that $u_0=f$ the formula 
	\[
	(u_\lambda)_\mu = u_{\lambda+\mu}
	\]
	holds for all $\lambda,\mu\geq 0$.
\end{proposition}
Here we have tweaked the notation slightly to make the statement more compact: By using the letter $u$ in place of $f$ for the insignal, the operation of denoising the signal for some $\lambda>0$ is indicated by adding the subscript `$\lambda$´ to the original signal $u$ thus obtaining $u_\lambda$. This makes sense even for $\lambda=0$ if we agree to set $u_0=u$.

A proof of the semi-group property can be found in \cite{Scherzer-etal-2009}. However, the fundamental estimate in  Theorem~\ref{thm:L-S-consequence} and the characterization of the ROF-minimizer in Theorem~\ref{thm:main} allow us to present short and very direct proof of this result:

\begin{proof}
	The assertion holds trivially if either $\lambda$ or $\mu$ equals zero, so we may assume that $\lambda, \mu > 0$. The idea of the proof is then to set $\bar{u} = (u_\lambda)_\mu$ and show that there exists a function $\bar{\xi}\in K$ such that
	\[
	\begin{cases}
	\bar{u} = f - (\lambda+\mu)\bar{\xi}'\quad\text{and }\\[2mm]
	J(\bar{u}) = \langle \bar{u} , \bar{\xi}' \rangle.
	\end{cases}.
	\]
	The characterization of solutions to the ROF model in Theorem~\ref{thm:main} then implies that $\bar{u}$ equals $u_{\lambda+\mu}$. Since $u_\lambda$ and $\bar{u}$ are the ROF-minimizers of $E_\lambda(f;\cdot)$ and $E_\mu(u_\lambda;\cdot)$, respectively, they both satisfy the conditions (\ref{eq:fundamental-rel}) and (\ref{eq:TV}), that is 
	\[
	\begin{cases}
	u_\lambda = f - \lambda \xi_\lambda'\; ,\\[2mm]
	J(u_\lambda) = \langle u_\lambda , \xi_\lambda'\rangle,
	\end{cases}
	\quad\text{and}\qquad
	\begin{cases}
	\bar{u} = u_\lambda - \mu\bar{\xi}_\mu'\quad\; ,\\[2mm]
	J(\bar{u}) = \langle \bar{u} , \bar{\xi}_\mu' \rangle,
	\end{cases}.
	\]
	for a uniquely determined pair of functions $\xi_\lambda$ and $\bar{\xi}_\mu$ in $K$.
	Now, if we set
	\[
	\bar{\xi} = \frac{\lambda\xi_\lambda+ \mu\bar{\xi}_\mu}{\lambda + \mu}
	\]
	then $\bar{\xi}\in K$ because it is the convex combination of two elements of $K$. Using what is known about $u_\lambda$ and $\bar{u}$, the following calculation reveals why we make this definition of $\bar{\xi}$:
	\begin{align*}
	f - (\lambda+\mu)\bar{\xi}' &= f - \lambda\xi_\lambda'-\mu\bar{\xi}_\mu \\[2mm]
	&=  u_\lambda -\mu\bar{\xi}_\mu \\[2mm]
	&= \bar{u},
	\end{align*}
	hence $\bar{u}$ and $\bar{\xi}$ fulfil the condition (\ref{eq:fundamental-rel}) by construction. It remains to verify that (\ref{eq:TV}) is fulfilled as well. Since
	\begin{align*}
	\langle \bar{u} , \bar{\xi}'\rangle & = \frac{\lambda}{\lambda+\mu}\langle \bar{u} , \xi_\lambda'\rangle + \frac{\mu}{\lambda+\mu}\langle \bar{u} , \bar{\xi}_\mu'\rangle \\[2mm]
	&=\frac{\lambda}{\lambda+\mu}\langle \bar{u} , \xi_\lambda'\rangle + \frac{\mu}{\lambda+\mu}J(\bar{u}) \, .
	\end{align*}
	we see that the second condition follows if it can show that $\langle \bar{u} , \xi_\lambda'\rangle = J(\bar{u})$. This essentially follows from the identity in Proposition~\ref{thm:strong-convergence} which states that $J(\bar{u})=J(u_\lambda) - J(u_\lambda-\bar{u})$. In fact, using this identity we get the inequality
	\begin{align*}
	J(\bar{u}) &\leq J(u_\lambda) - \langle u_\lambda -\bar{u} , \xi_\lambda'\rangle \\[2mm]
	&= J(u_\lambda) - J(u_\lambda) + \langle \bar{u} , \xi_\lambda'\rangle = \langle \bar{u} , \xi_\lambda'\rangle
	\end{align*}
	But $J(\bar{u}) \geq \langle \bar{u},\xi'\rangle$ for all $\xi\in K$, so $J(\bar{u}) = \langle \bar{u} ,\xi_\lambda'\rangle$, and the proof is complete.\qed
\end{proof}

The last part of the proof yields

\begin{corollary}
	If $\lambda > 0 0$ then $J(u_{\lambda}) = \langle u_{\lambda} , \xi_\mu' \rangle_{L^2(I)}$ for all $\mu$, $0<\mu\leq \lambda$.
\end{corollary}
Thus, in the computation of the total variation of $u_\lambda$ any of the previous $\xi_\mu$'s can be used.

\section{Application to Isotonic Regression}
We briefly outline how the theory developed earlier can be modified in order to derive the so-called ``lower convex envelope'' interpretation of the solution to the problem of isotonic regression. Isotonic regression is a method from mathematical statistics used for non-parametric estimation of probability distributions, see for instance Anevski and Soulier~\cite{Anevski-Soulier-2011}. It is a least-squares problem with a monotonicity constraints: given $f\in L^2(I)$, find the non-decreasing function $u_\uparrow\in L^2(I)$ which solves the minimization problem,
\begin{equation}\label{eq:IRdef}
\min_{u\in L^2_\uparrow(I)}\frac{1}{2}\|u-f\|_{L^2(I)}^2\;,
\end{equation}
where $L^2_\uparrow(I)$ denotes the set of all non-decreasing functions in $L^2(I)$.

\begin{figure*}
	\centering
	\subfigure[][The piecewise constant input signal {$f$} and the monotonic solution $u_\uparrow$ to the isotonic regression problem.]{
		\begin{tikzpicture}[scale=.575,>=stealth]
		\draw [thick,->] (-1,0) -- (7,0) node[anchor=north] {\small $x$};
		\draw [thick,->] (-.5,-2.5) -- (-.5,2.5) node[anchor=west] {\small $y$};
		\draw [thick] (0,-1) -- (.95,-1);
		\draw [fill] (0,-1) circle (.05);
		\draw (1,-1) circle (.05);
		\draw [thick] (1,-2) -- (1.95,-2);
		\draw [fill] (1,-2) circle (.05);
		\draw (2,-2) circle (.05);
		\draw [thick] (2,-.45) -- (2.95,-.45);
		\draw [fill] (2,-.45) circle (.05);
		\draw (3,-.45) circle (.05);
		\draw [thick] (3,1) -- (3.95,1);
		\draw [fill] (3,1) circle (.05);
		\draw (4,1) circle (.05);
		\draw [thick] (4,.5) -- (4.95,.5);
		\draw [fill] (4,.5) circle (.05);
		\draw (5,.5) circle (.05);
		\draw [thick] (5,2.05) -- (5.95,2.05);
		\draw [fill] (5,2.05) circle (.05);
		\draw [fill] (6,2.05) circle (.05) node [anchor=south west]{\small $f$};
		
		\draw [thick,orange] (0,-1.5) -- (1.95,-1.5);
		\draw [fill,orange] (0,-1.5) circle (.05);
		\draw [orange] (2,-1.5) circle (.05);
		\draw [thick,orange] (2,-.55) -- (2.95,-.55);
		\draw [fill,orange] (2,-.55) circle (.05);
		\draw [orange] (3,-.55) circle (.05);
		\draw [thick,orange] (3,.75) -- (4.95,.75);
		\draw [fill,orange] (3,.75) circle (.05);
		\draw [orange] (5,.75) circle (.05);
		\draw [thick,orange] (5,1.95) -- (6,1.95);
		\draw [fill,orange] (5,1.95) circle (.05);
		\draw [fill,orange] (6,1.95) circle (.05) node [anchor=north west]{\small $u_{\uparrow} = W_{\uparrow}'$};
		
		\draw [thick] (0,.1) -- (0,-.1) node[anchor=north] {\small $a$};
		\draw [thick] (1,.1) -- (1,-.1);
		\draw [thick] (2,.1) -- (2,-.1);
		\draw [thick] (3,.1) -- (3,-.1);
		\draw [thick] (4,.1) -- (4,-.1);
		\draw [thick] (5,.1) -- (5,-.1);
		\draw [thick] (6,.1) -- (6,-.1) node[anchor=north] {\small $b$};
		\end{tikzpicture}
	}
	\qquad\qquad
	\subfigure[][The cumulative signal {$F$} and the corresponding lower convex envelope (or taut string) $W_\uparrow$.]{
		\begin{tikzpicture}[scale=0.575,>=stealth]
		\draw [thick,->] (-1,0) -- (7,0) node[anchor=north] {\small $x$};
		\draw [thick,->] (-.5,-3.5) -- (-.5,1.5) node[anchor=west] {\small $y$};
		\draw [thick] (0,-.1) -- (0,.1) node[anchor=south] {\small $a$};
		\draw [thick] (1,.1) -- (1,-.1);
		\draw [thick] (2,.1) -- (2,-.1);
		\draw [thick] (3,.1) -- (3,-.1);
		\draw [thick] (4,.1) -- (4,-.1);
		\draw [thick] (5,.1) -- (5,-.1);
		\draw [thick] (6,-.1) -- (6,.1) node[anchor=south] {\small $b$};
		
		\draw [thick] (0,0) -- (1,-1) -- (2,-2.98) -- (3,-3.48) -- (4,-2.5) -- (4.95,-2) node [anchor=south east]{\small $F$} -- (5.95,0);
		\draw [thick,orange] (0,0) circle (.1);
		\draw [fill] (0,0) circle (.07);
		\draw [thick,orange] (5.975,0) circle (.1);
		\draw [fill] (5.975,0) circle (.07);
		
		\draw [thick, orange] (.02,-.08) -- (2,-3.05) -- (3,-3.55) -- (5,-2.05) node [anchor=north west,orange]{\small\bf $W_{\uparrow}$} -- (5.99,-.07);

		\end{tikzpicture}
	}
	
	\caption{A graphical illustrations of the taut string interpretation of isotonic regression.}\label{fig:IR}
\end{figure*}
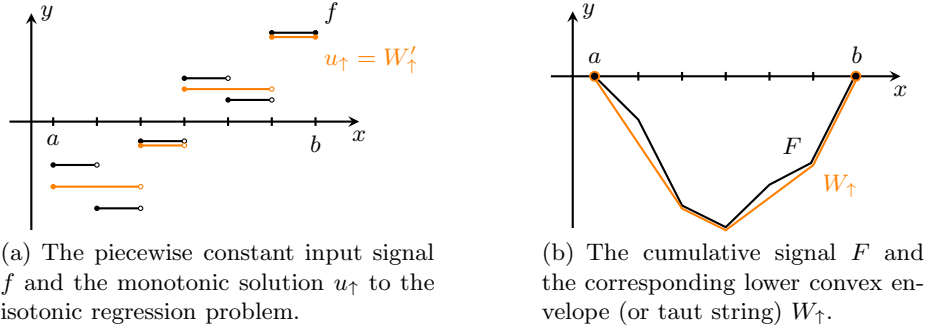

The idea is to re-formulate (\ref{eq:IRdef}) as an unconstrained optimization problem by replacing the total variation term $J$ of the ROF functional by regularization term $J_\uparrow$ which can distinguish between functions that are non-decreasing or not. To achieve this we set $K_+ = \big\{ \xi\in H_0^1(I)\, :\, \xi(x)\geq 0\text{ for all $x\in I$}  \big\}$ and define $J_\uparrow(u) = \sup_{\xi\in K_+} \, \langle u,\xi'\rangle_{L^2(I)}$. It can be shown that
\[
J_\uparrow(u) = 
\begin{cases}
0 & \text{if $u\in L_\uparrow^2(I)$}\;,\\
+\infty & \text{otherwise.}
\end{cases}
\]
The isotonic regression problem (\ref{eq:IRdef}) now becomes equivalent to finding the minimizer $u_\uparrow$ in $L^2(I)$ of the functional
\begin{equation}
\label{eq:newIR}
E_\uparrow(u) := J_\uparrow(u) + \frac{1}{2}\|u-f\|_{L^2(I)}^2.
\end{equation}
Notice that there is no need for a positive weight in this functional because the regularizer assumes only the values zero and infinity.

Again we may assume the mean value $f$ to be zero so that the cumulative function $F$ belongs to $H_0^1(I)$.
Mimicking the proof of Theorem~\ref{thm:main} we get:
\[
\min_{u\in L^2(I)} E_\uparrow(u) = \max_{W\in T}\frac{1}{2}\Big\{\|f\|^2 - \frac{1}{2}\| W'\|_{L^2(I)}^2\Big\}
\]
where $W=F-\xi$, $\xi\in K_+$, and $T=\{ W\in H_0^1(I)\, :\, W(x)\leq F(x), \text{ $x \in I$}\}$. The minimization of (\ref{eq:newIR}) is equivalent to the obstacle problem $\min_{W\in T} \frac{1}{2}\| W'\|_{L^2(I)}^2$ which admits a unique solution $W_\uparrow$ by the Projection theorem. It follows that (\ref{eq:newIR}) also has the unique solution $u_\uparrow = W_\uparrow'\quad(\text{distributional derivative})$ which belongs to $L_\uparrow^2(I)$ because $E_\uparrow(u_\uparrow)$ is finite.

The solution $W_\uparrow$ of the obstacle problem is automatically a convex function. In fact, by optimality, $W_\uparrow$ is the maximal convex function lying below $F$, i.e., it is the {\em lower convex envelope} of $F$. This interpretation is illustrated for a piecewise constant signal $f$ in Fig.~\ref{fig:IR}. Similar problems are considered in the multidimensional case, using higer-order methods (the space of functions with bounded Hessians), in Hinterberger and Scherzer~\cite{Hinterberger-Scherzer-2006}.

\subsubsection*{Acknowledgements} I want to thank Viktor Larsson at the Centre for Mathematical Sciences, Lund University, for reading and commenting the first draft of this paper.

\newpage
\appendix

\section{Proof of the ``Same Solution-Property''}\label{sec:appendix}

As promised in the introduction, we are going to prove that the solution of the minimization problem (\ref{eq:energy}) in {\sc step 3} of the Taut string algorithm coincides with the solution of the shortest path problem (\ref{eq:SPP}). In fact we prove the slightly more general statement:

\begin{lemma}\label{thm:same-solution-property}
	Let $H$ denote any strictly convex $C^1$-function defined on $\real$ and set
	\[
	L_H(W)=\int_I H(W'(x))\,dx\,.
	\]
	Then the problem $\min_{W\in T_\lambda} L_H(W)$ has precisely the same solution as the minimization problem $\min_{W\in T_\lambda}E(W)$ in (\ref{eq:energy}).
\end{lemma}
The original problem is then solved by taking $H(s)=(1+s^2)^{1/2}$.

\begin{proof}
The idea of the proof is to verify that $W_\lambda:=\operatorname{arg\,min}_{W\in T_\lambda}E(W)$ solves the variational inequality:
\begin{equation}\label{eq:VI-for-SPP}
\int_I h(W_\lambda'(x)) (W'(x)-W_\lambda(x)')\, dx \geq 0\, , \text{ for all $W\in T_\lambda$},
\end{equation}
where $h=H'$. This condition is both necessary and sufficient for $W_\lambda$ to be a minimizer of $L_H$ over $T_\lambda$, and since $L_H$ is a strictly convex functional, there is at most one such minimizer.

Being the minimizer of $E$ over $T_\lambda$, $W_\lambda\in T_\lambda$ satisfies the variational inequality (which is a special case of (\ref{eq:VI-for-SPP}) if we take $H(s)=s$):
\begin{equation}\label{eq:VI}
\int_I W_\lambda'(W'-W_\lambda')\,dx \geq 0,\text{ for all $W\in T_\lambda$.}
\end{equation}
Set $C_+=\{ x\in I : W_\lambda(x)=F(x)+\lambda\}$ and $C_-=\{ x\in I : W_\lambda(x)=F(x)-\lambda\}$. These are the sets where the solution touches the upper and the lower obstacles, respectively. Since $F$ and $W_\lambda$ are continuous, both  sets are closed. In fact, $C_+$ and $C_-$ are compact because $\lambda>0$ implies that they do not reach the boundary of $I$. They are disjoint, $C_+\cap C_-=\emptyset$, and their union, $C=C_+\cup C_-$, is the contact set for $W_\lambda$. 

For any non-negative $\xi\in C_0^1(I\backslash C_+)$ there exists an $\epsilon >0$ such that $W:=W_\lambda + \epsilon\xi$ belongs to $T_\lambda$. If this $W$ is substituted into (\ref{eq:VI}) we find that
\[
\int_I W_\lambda'\xi'\,dx\geq 0\quad\text{for all $\xi\in C_0^1(I\backslash C_+)$ with $\xi\geq 0.$}
\]
It follows that $-W_\lambda''$ is a positive measure on $I\backslash C_+$, hence $-W_\lambda'$ is non-decreasing on each connected component of $I\backslash C_+$. Similarly one proves that $-W_\lambda'$ is non-increasing on each connected component of $I\backslash C_-$. This means, in particular, that $W_\lambda'$ constant on each connected component of $I\backslash C$.

Since $h$ is non-decreasing, the composite function $-h(W_\lambda')$ has the same monotonicity properties as $-W_\lambda'$. Therefore the distributional derivative $-h(W_\lambda')'$ is a positive measure $\mu^+$ on $I\backslash C_+$ and minus a positive measure $-\mu^-$ on $I\backslash C_-$. Clearly $\operatorname{supp} \mu^+\subset C_-$ and $\operatorname{supp} \mu^-\subset C_+$, so $-h(W_\lambda')'$ is a signed measure $\mu$ with the Jordan decomposition $\mu=\mu^+ -\mu^-$. The following calculation now verifies (\ref{eq:VI-for-SPP}): For any $W\in T_\lambda$ we have
\begin{align*}
\int_I h(W_\lambda')(W'-W_\lambda')\,dx &= \int_I W-W_\lambda\, d\mu \\
&= \int_I W-W_\lambda\,d\mu^+ - \int_I W-W_\lambda\,d\mu^- \geq 0
\end{align*}
which holds because $ W-W_\lambda \geq 0$ on $C_-$ and $ W-W_\lambda\leq 0$ on $C_+$.\qed
\end{proof}

\end{document}